\documentclass[manuscript,nonacm,10pt]{acmart} 

\AtBeginDocument{%
  \providecommand\BibTeX{{%
    \normalfont B\kern-0.5em{\scshape i\kern-0.25em b}\kern-0.8em\TeX}}}

\copyrightyear{2021}
\acmYear{2021}
\setcopyright{acmlicensed}\acmConference[PODC '21]{Proceedings of the 2021 ACM Symposium on Principles of Distributed Computing}{July 26--30, 2021}{Virtual Event, Italy}
\acmBooktitle{Proceedings of the 2021 ACM Symposium on Principles of Distributed Computing (PODC '21), July 26--30, 2021, Virtual Event, Italy}
\acmPrice{15.00}
\acmDOI{10.1145/3465084.3467938}
\acmISBN{978-1-4503-8548-0/21/07}

\usepackage{graphicx}

\usepackage{color}				

\usepackage{algorithm}
\usepackage[noend]{algpseudocode}
\usepackage{caption}

\makeatletter
\newenvironment{algoFigure}[1][htb]{%
		\let\c@algorithm\c@figure
    \renewcommand{\ALG@name}{Figure}
    \begin{algorithm}[#1]%
    }{\end{algorithm}
}
\makeatother

\usepackage{aliascnt}
\usepackage{hyperref}

\newtheorem{theorem}{Theorem}

\newaliascnt{lem}{theorem}  
  
\aliascntresetthe{lem}

\newaliascnt{corol}{theorem}  
  
\aliascntresetthe{corol}

\newaliascnt{observation}{theorem}  
\newtheorem{observation}[observation]{Observation}  
\aliascntresetthe{observation}

\newtheorem{clm}{Claim}[theorem]

\newaliascnt{observationSub}{clm}  
\newtheorem{observationSub}[observationSub]{Observation}  
\aliascntresetthe{observationSub}
  
\newtheorem{subclm}{Subclaim}[clm]

\newenvironment{temptheoremcounter}[1]
	{\par\edef\savedtheoremnumber{\number\value{theorem}}
		\setcounterref{theorem}{#1}
		\addtocounter{theorem}{-1}}
	{\par\setcounter{theorem}{\savedtheoremnumber}}

\newcommand{\MP}[1]{{\color{red}\rule{1pt}{5pt}}\marginpar{\tiny {#1}}}

\newcommand{\RMP}[1]{}

\newcommand{\BMP}[1]{}

\usepackage{amsmath}

\definecolor{purple}{rgb}{0.6,0,0.6}

\definecolor{darkgreen}{rgb}{0,0.4,0}

\newcommand{\rmv}[1]{}

\newcommand{\vDone}{\textbf{done}}

\newcommand{\oread}{\textsc{read}}
\newcommand{\otest}{\textsc{test}}
\newcommand{\oset}{\textsc{set}}

\newcommand{\qH}[1]{\qed}

\newcommand{\seth}{\mathcal{H}}

\newcommand{\CDeli}{C_\textit{deli}}
\newcommand{\HBi}{H_\textit{bi}}

\newcommand{\str}{\textrm{write strong-linearization}} 
\newcommand{\sly}{\textrm{write strongly-linearizable}}

\newcommand{\CInit}{C_{{\textit{init}}}}

\begin{document}

\title{On Strong Linearizability and Write Strong Linearizability
in Message-Passing}

\author{David Yu Cheng Chan}
\email{david.chan1@ucalgary.ca}
\affiliation{%
	\institution{University of Calgary}
	\city{Calgary}
	\state{Alberta}
	\country{Canada}%
}

\author{Vassos Hadzilacos}
\email{vassos@cs.toronto.edu}
\affiliation{%
	\institution{University of Toronto}
	\city{Toronto}
	\state{Ontario}
	\country{Canada}%
}

\author{Xing Hu}
\email{xing@cs.toronto.edu}
\affiliation{%
	\institution{University of Toronto}
	\city{Toronto}
	\state{Ontario}
	\country{Canada}%
}

\author{Sam Toueg}
\email{sam@cs.toronto.edu}
\affiliation{%
	\institution{University of Toronto}
	\city{Toronto}
	\state{Ontario}
	\country{Canada}%
}
 

\begin{abstract}

We prove that in asynchronous message-passing systems where at most one process may crash,
	there is no lock-free
	strongly linearizable implementation of a weak object that we call Test-\emph{or}-Set (ToS).
	This object allows a single distinguished process to apply the $\oset$ operation once,
and a \emph{different} distinguished process to apply the $\otest$ operation also once.
Since this weak object can be directly implemented by a single-writer single-reader (SWSR) register
	(and other common objects such as max-register,  snapshot and counter),
	this result implies that there is no $1$-resilient lock-free strongly linearizable implementation
	of a SWSR register (and of these other objects) in message-passing systems.

We also prove that there is no $1$-resilient lock-free \emph{write} strongly-linearizable implementation of a 2-writer 1-reader (2W1R) register in asynchronous message-passing systems.

\end{abstract}

\maketitle

\section{Introduction}
In seminal work, Golab, Higham, and Woelfel showed that linearizability has the following limitation: 
	a \emph{randomized} algorithm that works with atomic objects against a strong adversary
	may lose some of its properties if we replace the atomic objects that it uses with objects that are only linearizable~\cite{sl11}.
To address this, they proposed a stronger version of linearizability, called \emph{strong linearizability}:
	intuitively, while in linearizability the order of all operations can be determined ``off-line'' given the entire execution,
	in strong linearizability the order of all operations has to be fixed irrevocably ``on-line'' without knowing the rest of the execution.
Golab \emph{et al.} proved that strongly linearizable (implementations of) objects are
	``as good'' as atomic objects for randomized algorithms against a strong adversary:
	they can replace atomic objects while preserving the algorithm's correctness properties.
 
There are many cases, however, where strong linearizability is impossible to achieve.
In particular, for shared-memory systems, Helmi \emph{et al.} proved that
	a large class of so-called \emph{non-trivial} objects including multi-writer registers, max-registers, snapshots, and counters,
	do not have strongly linearizable non-blocking
	implementations from single-writer multi-reader (SWMR) registers~\cite{sl12},



%
In this paper, we consider asynchronous message-passing systems where at most one process may crash,
	and show that in such systems
	there is no lock-free strongly linearizable implementation of
	a weak object that we call \emph{Test-or-Set (ToS)}:\footnote{Henceforth,
	we say that an object implementation is \emph{$1$-resilient lock-free}
	if it is lock-free under the assumption that at most one process may crash.
	This progress condition is defined more precisely in Section~\ref{Objects}.}
	with ToS, 
	one distinguished process can apply the $\oset$ operation once,
	and another distinguished process can apply the $\otest$ operation also only once:
	the $\otest$ operation returns $1$ if $\oset$ has previously been applied, and it returns $0$ otherwise.

Since a single-writer single-reader (SWSR) register directly implements a ToS object, the above result immediately
	implies that there is no $1$-resilient lock-free
	strongly linearizable implementation of a
	SWSR register in message-passing systems.
This result strengthens a recent result by Attiya, Enea, and Welch
	which shows that \emph{multi}-writer registers do not have strongly linearizable \emph{nonblocking}
	implementations
	in message-passing systems~\cite{AEW21},
	\footnote{Appendix~\ref{appendix} shows that the nonblocking progress condition defined in~\cite{AEW21} implies $1$-resilient lock-freedom.}
	and it also answers an open question asked in that paper,
	namely, whether there is a (fault-tolerant) strongly linearizable implementation
	of a \emph{single}-writer register in message-passing systems.
Attiya \emph{et al.} also prove that max-registers, snapshot, and counters objects do not have
	nonblocking strongly linearizable implementations in message-passing systems.
Since max-registers, snapshot, and counters also directly implement a ToS object, our result implies
	that there is no $1$-resilient lock-free strongly linearizable implementation
	of these objects in message-passing systems.
%
%

Since there is no nonblocking strongly linearizable
	implementation of MWMR registers from SWMR registers
	in shared-memory systems~\cite{sl12},
	Hadzilacos, Xing, and Toueg proposed an \emph{intermediate} notion of linearizability, called \emph{write strong-linearizability}~\cite{HHT211},
	and proved that
	there \emph{is} a wait-free write strongly-linearizable implementation of MWMR registers from SWMR registers.
This is useful because some randomized algorithms that do \emph{not} terminate with linearizable MWMR registers,
	terminate with write strongly-linearizable MWMR registers~\cite{HHT211}.

The above results raise the following question:
	is there a \emph{write} strongly-linearizable implementation of MWMR registers in message-passing systems?
We prove here that the answer is negative.
More precisely, we prove that there is no
	 $1$-resilient lock-free write strongly-linearizable implementation of even a 2-Writer 1-Reader (2W1R) register in message-passing systems.

\section{Model sketch}
We consider a standard asynchronous message-passing distributed system where processes communicate via messages and may fail by crashing.
The proof of our result is based on a bivalency argument, and we assume that the reader is familiar with the model and terminology
	introduced by Fischer, Lynch, and Paterson in~\cite{FischerEtal1985}
	to prove their famous impossibility result.
Recall that in this model, processes take \emph{steps}, and in each step a process does the following:
	it attempts to receive a message $m$ previously sent to it  ($m=\bot$ if does not receive any message),
	it changes state according to the message received, and it sends a finite set of messages to other processes.
A step taken by a process $p$ in which it receives message $m$ is denoted $e=(p,m)$. 
A configuration consists of the state of all the processes and of the message buffer
	(which consists of all the messages sent but not yet received).
A history is a sequence of steps.
A process is \emph{correct} if it takes infinitely many steps,
	we say that it \emph{crashes} otherwise.
We consider systems with reliable communication links, i.e., where every message sent to a correct process is eventually received. 
Henceforth we consider only such executions.


\subsection{Objects}\label{Objects} 
We consider implementations of ToS and 2W1R objects in message-passing systems.
In an object implementation,
  each operation spans an interval that starts
  with an \emph{invocation} and terminates with a \emph{response}.
For any two operations $o$ and $o'$, $o$ \emph{precedes} $o'$ if the response of $o$ occurs 
  before the invocation of $o'$, and $o$ \emph{is concurrent with} $o'$ if neither precedes the other.

Roughly speaking,
  an object implementation is \emph{linearizable}~\cite{HerlihyWing1990}
  if operations (which may be concurrent)
  behave as if they occur in some sequential order (called ``linearization order'')
  that is consistent with the order in which operations actually occur:
  if an operation $o$ precedes an operation $o'$, then $o$ is before $o'$ in the linearization order.

In the message-passing model,
	the occurrence of an invocation or response to an operation
	is encoded by a change in the state of the process that invoked the operation (and by a corresponding change in the configuration).
Let $\seth$ be the set of all histories of an object implementation (note that this set is prefix-closed).
A history $H \in \seth$ is \emph{sequential} if no two operations are concurrent in $H$.\footnote{Strictly-speaking, what we mean here is that ``no operations are concurrent'' in the configurations reached when applying $H$ to the initial configuration of this object implementation. For convenience, when clear, we sometimes use $H$ to denote the sequence of configurations obtained  by applying $H$ to the initial configuration of an object.}
An operation $o$ is \emph{complete} in a history $H\in\seth$ if $H$ contains both the invocation and response of $o$,
 	otherwise $o$ is \emph{pending}.
A \emph{completion} of a history $H$ is a history $H'$ obtained from $H$ by removing a subset of the pending operations
	and completing the remaining ones with responses.

Let $\seth$ be the set of all histories of an implementation of an object of type $T$.\footnote{Intuitively, an object type is specified by how this object behaves when it is accessed sequentially~\cite{HerlihyWing1990}.}

\begin{definition}\label{def-T}
A function $f$ is \emph{a linearization function for $\seth$} (with respect to type $T$)
  if it maps each history $H \in \seth$ to a sequential history $f(H)$ such that:
\begin{enumerate}
  \item \label{Tp1} $f(H)$ has exactly the same operations as some completion $H'$ of $H$.

  \item \label{Tp2} If operation $o$ precedes $o'$ in $H'$,
    then $o$ occurs before $o'$ in $f(H)$.

  \item \label{Tp3} The sequence of operations in $f(H)$ conforms to the type $T$. 
\end{enumerate}
\end{definition}

\begin{definition}\label{L-T}
An algorithm $I$ that implements an object of type $T$
  is \emph{linearizable}
  if there is a linearization function (with respect to type $T$)
  for the set of histories~$\seth$~of~$I$.
\end{definition}

We now define a weak progress condition for object implementations in message-passing systems.
Intuitively, a $1$-resilient lock-free implementation guarantees that if at most one process crashes, then the implementation is lock-free:
	whenever a correct process has a pending operation, some operation will complete.
More precisely:

%

\begin{definition}[1-resilient lock-free]\label{OneResilient}
An implementation is $1$-resilient lock-free if it satisfies the following property.
Let $C$ be any reachable configuration that has a pending operation by some process $p$.
For every infinite history $H$
	that is applicable to $C$
	such that at most one process crashes and $p$ is correct,
	there is a finite prefix $H'$ of $H$ such that more operations have completed in $H'(C)$ than in $C$.
\end{definition}

\subsection{Test-or-Set}

\newcommand{\dBit}{\textsc{b}}


We define a simple object called \emph{Test-or-Set (ToS)} as follows.
The state of a ToS object is a single bit $b$, initially $0$.
A~single distinguished process is allowed to apply a single operation 
	$\otest$ that returns the value of $b$; it can apply this operation only once.
A different distinguished process is allowed to apply a single operation
	$\oset$ that~sets~$b$~to~$1$ (and returns $\vDone$); it can apply this operation only once.
The sequential specification of the ToS object is the obvious one:
	the $\otest$ operation returns $1$
	if a $\oset$ operation has previously been applied,
	and $0$ if not.

\subsection{Strongly Linearizable Implementations of Test-or-Set}

In this section, $\seth$ is the set of all histories of a ToS implementation.

\begin{definition}\label{def-ToS}
A function $f$ is \emph{a linearization function for $\seth$} (with respect to the type ToS)
  if it maps each history $H \in \seth$ to a sequential history $f(H)$ such that:
\begin{enumerate}
  \item \label{ToSp1} $f(H)$ has exactly the same operations as some completion $H'$ of $H$.

  \item \label{ToSp2} If operation $o$ precedes $o'$ in $H'$,
    then $o$ occurs before $o'$ in $f(H)$.

  \item \label{ToSp3} For any $\otest$ operation $t$ in $f(H)$,
    if no $\oset$ operation occurs before $t$ in $f(H)$, then $t$ reads 0;
    otherwise, $t$ reads~1.
\end{enumerate}
\end{definition}

\begin{definition}\cite{sl11}\label{def:SL}
A function $f$ is \emph{a strong linearization function for $\seth$} if:

(L) $f$ is a linearization function for $\seth$, and
  
(P) for any histories $G,H \in \seth$, if $G$ is a prefix of $H$,
  then $f(G)$ is a prefix of $f(H)$.

\end{definition}

\begin{definition}\label{L-ToS}
An algorithm $I$
	that implements a \emph{ToS} object
  is \emph{strongly linearizable}
  if there is a strong linearization function (with respect to the type ToS)
  for the set of histories~$\seth$~of~$I$.
\end{definition}

%
%
%

\subsection{Write Strongly-Linearizable Implementations of Registers}

In this section, $\seth$ is the set of all histories of a register implementation.

\begin{definition}\label{def-Reg}
A function $f$ is \emph{a linearization function for $\seth$} (with respect to the type register)
  if it maps each history $H \in \seth$ to a sequential history $f(H)$ such that:
\begin{enumerate}
  \item \label{Rp1} $f(H)$ has exactly the same operations as some completion $H'$ of $H$.

  \item \label{Rp2} If operation $o$ precedes $o'$ in $H'$,
    then $o$ occurs before $o'$ in $f(H)$.

  \item \label{Rp3} For any read operation $r$ in $f(H)$,
    if no write operation occurs before $r$ in $f(H)$, then $r$ reads the initial value of the register; 
    otherwise, $r$ reads the value written by the last write operation that occurs before $r$ in $f(H)$.
\end{enumerate}
\end{definition}

\begin{definition}\cite{HHT211}\label{def:WSL}
A function $f$ is \emph{a $\str$ function for $\seth$} if:

(L) $f$ is a linearization function for $\seth$, and

(P) for any histories $G,H \in \seth$, if $G$ is a prefix of $H$, 
  then the sequence of write operations in $f(G)$ is a prefix of the sequence of write operations in $f(H)$.
\end{definition}

An algorithm $I$
	that implements a \emph{ToS} object
  is \emph{strongly linearizable}
  if there is a strong linearization function (with respect to the type ToS)
  for the set of histories~$\seth$~of~$I$.

\begin{definition}\label{L-Reag}
An algorithm $I$
	that implements a \emph{register}
  is \emph{$\sly$}
  if there is a write strong-linearization function (with respect to the type register)
  for the set of histories~$\seth$~of~$I$.
\end{definition}

\section{Impossibility Results}


\subsection{Test-or-Set}

\begin{theorem}
	\label{thm:ToS}
	For all $n \ge 2$, there is no $1$-resilient lock-free strongly linearizable
		implementation of the Test-or-Set object 
		in an asynchronous message-passing system of $n$ processes. 
\end{theorem}

\begin{proof}
	Suppose, for contradiction, that 
		there is a $1$-resilient lock-free strongly linearizable
		implementation $I$ of the Test-or-Set object 
		in a message-passing system of $n$ processes. 
	Without loss of generality,
		suppose that this implementation 
		only allows process $0$ to invoke the $\otest$ operation
		and only allows process $1$ to invoke the $\oset$ operation.
	The other $n-2$ processes may take steps to help complete operations,
		but do not invoke any operations themselves. 

Let $A$ be the message-passing algorithm that uses the implementation $I$
	as shown in Figure~\ref{algo:ToSA}. 

\begin{algoFigure}[t]
	\caption{Algorithm $A$ that uses the implementation $I$ of a test-or-set object}\label{algo:ToSA}
	\begin{algorithmic}[1]
		\Statex Code for process $0$:
		\State $I.\otest$ 
		\Statex
		\Statex Code for process $1$:
		\State $I.\oset$  
	\end{algorithmic} 
\end{algoFigure}

Henceforth, we consider the set of runs (i.e., histories) of algorithm $A$ in a message-passing system with $n \ge 2$ processes in which at most one process crashes
	and every message sent to a correct process is eventually received.

Let $\CInit$ be the initial configuration of algorithm $A$.
A configuration $C$ is \emph{reachable} if there is a history $H$
	such that $C=H(\CInit)$.
For any reachable configuration $C$:
	\begin{itemize}
	\item $C$ is \emph{$v$-valent} for $v \in \{0,1\}$ 
		if there is no finite history $H$ applicable to $C$ such that 
		the $\otest$ operation has returned $1-v$ in $H(C)$.
	\item $C$ is \emph{bivalent} if it is neither $0$-valent nor $1$-valent.  
	\item $C$ is \emph{univalent} if it is 
		\emph{$v$-valent} for exactly one $v \in \{0,1\}$.
	\end{itemize}

	\begin{observationSub}
		\label{obs:ToSVal}
		If the $\otest$ operation has returned some value $v \in \{0,1\}$ in some
		reachable configuration $C$, then $C$ is $v$-valent.
	\end{observationSub}

\begin{proof}
This holds
 because algorithm $A$ invokes at most one $\otest$ operation.
\end{proof}

	\begin{observationSub}
		\label{obs:ToSBoV}
		A reachable configuration $C$ cannot be both $0$-valent and $1$-valent. 
	\end{observationSub}
	
	\begin{proof}
		Consider the following history: starting from $C$, 
			every process takes steps in a round-robin order, and
			in each step a process receives the earliest pending message for it.
		So all processes take infinitely many steps
			and every message is eventually received.
		Since the implementation $I$ is $1$-resilient lock-free and no process crashes,
			it is clear that
			eventually 
			both processes $0$ and $1$ complete their execution of algorithm $A$;
			in particular, process $0$ completes its $\otest$ operation
			on the test-or-set object implemented by~$I$.
		So there exists a finite history $H$ applicable to $C$ such that
			the $\otest$ operation by process $0$ has completed in $H(C)$.
		Clearly, this $\otest$ operation returns some value $v \in \{0,1\}$.
		Thus $C$ is not $(1-v)$-valent.
	\end{proof}
		
	\begin{observationSub}
		\label{obs:ToSBiV}
		If $C$ is a reachable bivalent configuration, then for each $v \in \{0,1\}$, 
			there exists a finite history $H_v$ such that $H_v(C)$ is $v$-valent.
	\end{observationSub}
	
	\begin{proof}
		Let $C$ be a reachable bivalent configuration.
		Thus $C$ is neither $0$-valent nor $1$-valent.
		Thus for each $v \in \{0,1\}$,
			there exists a finite history $H_v$ applicable to $C$ such that
			the $\otest$ operation has returned $v$~in~$H_v(C)$.
		By \autoref{obs:ToSVal}, $H_v(C)$ is $v$-valent.
	\end{proof}
		
	\begin{clm}
		\label{thm:ToSInit}
		The initial configuration $\CInit$ 
			(before any operation is invoked) is bivalent.
	\end{clm}
	
	\begin{proof} 
		Consider the following run of $A$: 
			(a) process $1$ completes its $\oset$ operation;
			then (b) process $0$ invokes its $\otest$ operation and completes it.
		Clearly, the $\otest$ operation returns $1$.
		So $\CInit$ is not $0$-valent.
		
		Consider the following run of $A$:
			(a) process $0$ completes its $\otest$ operation;
			then (b) process $1$ invokes its $\oset$ operation, and completes it.
		Clearly, the $\otest$ operation returns $0$.
		So $\CInit$ is not $1$-valent.

		Since $\CInit$ is neither $0$-valent nor $1$-valent, it is bivalent.
	\end{proof} 
	
	\begin{clm}
		\label{thm:ToSCompleteMeansUni}
		If $C$ is a reachable bivalent configuration,
		then no operation has completed in $C$.
	\end{clm}

	\begin{proof}
		Suppose, for contradiction,
		that there is a history $H$ applicable to $\CInit$
			such that $C=H(\CInit)$ is bivalent,
			yet some operation has completed in $H(\CInit)$. 
		Suppose the $\otest$ operation has completed in $H(\CInit)$,
			and let $v \in \{0,1\}$  be the value returned.
		Then, by \autoref{obs:ToSVal}, $H(\CInit)$ is $v$-valent --- contradicting that $H(\CInit)$ is bivalent.
		So a $\oset$ operation has completed in $H(\CInit)$.
		
		Since $H(\CInit)$ is bivalent, it is neither $0$-valent nor $1$-valent.
		Thus for each $v \in \{0,1\}$,
			there exists a finite history $H_v$ applicable to $H(\CInit)$ such that
			the $\otest$ operation has returned $v$ in $H H_v(\CInit)$.
		For each $v \in \{0,1\}$, let $H'_v = H H_v$.
		
		Let $f$ be any strong-linearization function 
			for the set of all histories of
			the implementation $I$ of the test-or-set object.
		Then:
		\begin{itemize}
			\item In $f(H'_0)$, the $\otest$ operation 
				occurs before the $\oset$ operation
				because otherwise the $\otest$ operation could not return $0$.
			\item In $f(H'_1)$, the $\otest$ operation 
				occurs after the $\oset$ operation
				because otherwise the $\otest$ operation could not return $1$. 
		\end{itemize}  
		
		Since $H$ is a prefix of both $H'_0$ and $H'_1$,
			by Definition~\ref{def:SL},
			$f(H)$ is a prefix of both $f(H'_0)$ and $f(H'_1)$.
		So $f(H)$ does not contain any operations.
		However, since a $\oset$ operation has completed in $H(\CInit)$,
			by Definition~\ref{def-ToS} and~\ref{def:SL}
			$f(H)$ must contain this operation --- a contradiction.
	\end{proof}
	  
	\begin{clm}
	\label{thm:ToSNoDeli}
		Let $C$ be any reachable bivalent configuration,
			$e = (p, m)$ be any step that is applicable to $C$,
			$\mathcal{C}_{-e}$ be the set of configurations 
			reachable from $C$ without applying $e$,
			and $\mathcal{D} = e(\mathcal{C}_{-e}) = 
			\{e(E) | E \in \mathcal{C}_{-e} \textrm{ and } 
			e \textrm{ is applicable to } E\}$.
		Then $\mathcal{D}$ contains a bivalent configuration.
	\end{clm}
	
	\begin{proof}
		Note that since $e$ is applicable to $C$, 
			$e$ is also applicable to every configuration in $\mathcal{C}_{-e}$.
			
		Suppose, for contradiction, that 
			$\mathcal{D}$ does not contain a bivalent configuration.
		Then by \autoref{obs:ToSBoV}, $\mathcal{D}$ 
			contains only univalent configurations.
		Since $C$ is bivalent, by \autoref{obs:ToSBiV}, 
			for each $v \in \{0,1\}$ there is a configuration $E_v$
			reachable from $C$ such that $E_v$ is $v$-valent.
		For each $v \in \{0,1\}$,
			if $E_v$ is in $\mathcal{C}_{-e}$,
			then let $E'_v = e(E_v) \in \mathcal{D}$;
			otherwise $e$ was applied to reach $E_v$,
			so let $E'_v$ be a configuration in $\mathcal{D}$
			from which $E_v$ is reachable.
		Thus for each $v \in \{0,1\}$,
			$E'_v$ is in $\mathcal{D}$ and is $v$-valent.
		So $\mathcal{D}$ contains both $0$-valent and $1$-valent configurations.
		
		Consequently, there exists a configuration $\CDeli$ in $\mathcal{C}_{-e}$
			and a step $e' \neq e$ applicable to $\CDeli$ such that
			$e(\CDeli)$ and $e'e(\CDeli)$ are univalent configurations in $\mathcal{D}$
			that have opposite valence. 
		There are two cases: either $e$ and $e'$ are steps of different processes, 
			or $e$ and $e'$ are steps of the same process $p$.
		\begin{description}
			\item[\textbf{Case 1:}] $e$ and $e'$ are steps of different processes.
			
				Then $ee'(\CDeli) = e'e(\CDeli)$.
				Thus, $e'e(\CDeli)$ is both $0$-valent and $1$-valent 
					--- contradicting \autoref{obs:ToSBoV}.
							
			\item[\textbf{Case 2:}] $e$ and $e'$ are steps of the same process $p$.
				Consider the following history: starting from $\CDeli$, 
				every process except $p$ takes steps in a round-robin order, and
			        in each step a process receives the earliest pending message for it.
				So all processes except $p$ take infinitely many steps
				and every message is eventually received.
				Since the implementation $I$ is $1$-resilient lock-free,
				it is clear that
				eventually a configuration is reached where some operation has completed.
		Thus, there is a finite $p$-free history $H_{-p}$ applicable to $\CDeli$
					such that an operation has completed in $H_{-p}(\CDeli)$.
				Since $H_{-p}$ is $p$-free,
					$H_{-p}$ is applicable to both $e(\CDeli)$ and $e'e(\CDeli)$,
				and both $e$ and $e'e$ are applicable to $H_{-p}(\CDeli)$.
				Furthermore, it is clear that
					$eH_{-p} (\CDeli) = H_{-p}e(\CDeli)$ and
					$e'eH_{-p} (\CDeli) = H_{-p}e'e(\CDeli)$
					are univalent configurations with opposite valence.
				So $H_{-p}(\CDeli)$ is bivalent ---
					contradicting \autoref{thm:ToSCompleteMeansUni}
					since an operation has completed in $H_{-p}(\CDeli)$.
		\end{description}
		\vspace*{-3mm}
	\end{proof} 
	
	We now construct an infinite history $\HBi$ applicable to $\CInit$
		such that every configuration in $\HBi(\CInit)$ is bivalent
		as follows:
	\begin{enumerate}
		\item Let $C = \CInit$ ($C$ is bivalent by \autoref{thm:ToSInit}) and 
			let $S$ be an arbitrary sequence of all $n$ processes.
		\item Let $p$ be the first process in $S$, 
			$m$ be the earliest message in the message buffer of $p$ in $C$ 
			(or $\bot$ if no such message exists), and $e = (p,m)$.
		\item By \autoref{thm:ToSNoDeli},
			there is a bivalent configuration $C'$ reachable from $C$
			where $e$ has been applied.
		\item Move $p$ to the end of $S$,
			and let $C = C'$.
		\item Repeat from Step 2.
	\end{enumerate}

	Note that all the configurations ``traversed'' by 
		the application of $\HBi$ to $\CInit$ are bivalent.
	By \autoref{thm:ToSCompleteMeansUni}, no operation has completed in any of them.
	However, in the infinite history $\HBi$, 
		every process takes infinitely many steps and every message is eventually received.
	Thus, since the implementation $I$ of the test-or-set object 
		(used by algorithm $A$) is $1$-resilient lock-free,
		at least one operation by process $0$ or $1$ 
		on this test-or-set object must have completed
		--- a contradiction.
\end{proof}

Since $ToS$ can be \emph{directly} implemented by a SWSR register
	(and also by other objects such as max-registers, snapshot, and counters)
Theorem~\ref{thm:ToS} implies:

\begin{corollary}
	\label{cor:OO}
	For all $n \ge 2$, there is no $1$-resilient lock-free strongly linearizable
		implementation of
		SWSR registers in an asynchronous message-passing system of $n$ processes.
		Moreover, there is no such an implementation of
		max-registers, snapshot, and counters.
\end{corollary}

%

The above corollary is stronger than a result
	in~\cite{AEW21} which states that, if three or more processes
	are allowed to invoke operations,\footnote{These processes are called \emph{clients} in~\cite{AEW21}.}
	there is no strongly linearizable nonblocking message-passing
	implementation of multi-writer registers, max-registers, counters, or snapshot objects.
Note that relating the two results is not immediate, because Corollary~\ref{cor:OO}
	is about $1$-resilient lock-free implementations, while the result in~\cite{AEW21} is about nonblocking
	implementation as defined in~\cite{AEW21}.
In Appendix~\ref{appendix}, we show that nonblocking implies $1$-resilient lock-freedom.


\subsection{2W1R 1-bit Registers}

\begin{theorem}
	\label{thm:2W1R}
	For all $n \ge 2$, there is no $1$-resilient lock-free write strongly-linearizable
		implementation of a $1$-bit 2W1R register
		in an asynchronous message-passing system of $n$ processes.
\end{theorem}

\begin{proof}
	Suppose, for contradiction, that 
		there is a $1$-resilient lock-free write strongly-linearizable
		implementation $I$ of a $1$-bit 2W1R register 
		in a message-passing system of $n$ processes.
	Without loss of generality,
		suppose that this implementation 
		has initial value $0$,
		and only allows process $0$ to invoke read operations
		and processes $0$ and $1$ to invoke write operations.
	The other $n-2$ processes may take steps to help complete operations,
		but do not invoke any operations themselves. 

Let $A$ be the message-passing algorithm that uses the implementation $I$
	as shown in Figure~\ref{algo:A}.
Note that in algorithm~$A$, process $0$ does not invoke its read operation before process $1$ completes its write.

\begin{algoFigure}[t]
	\caption{Algorithm $A$ that uses the implementation $I$ of a 2W1R register}\label{algo:A}
	\begin{algorithmic}[1]
		\Statex Code for process $0$:
		\State $I.\textsc{write}(0)$
		\Repeat
		\Until{\textsc{receive} ``OK'' from process $1$}
		\State $I.\oread$
		\Statex
		\Statex Code for process $1$:
		\State $I.\textsc{write}(1)$ 
		\State \textsc{send} ``OK'' to process $0$
	\end{algorithmic} 
\end{algoFigure}

Henceforth, we consider the set of runs (i.e., histories) of algorithm $A$ in a message-passing system with $n \ge 2$ processes in which at most one process crashes
		and every message sent to a correct process is eventually received.

Let $\CInit$ be the initial configuration of algorithm $A$.
A configuration $C$ is \emph{reachable} if there is a history $H$
	such that $C=H(\CInit)$.
For any reachable configuration $C$:
	\begin{itemize}
	\item $C$ is \emph{$v$-valent} for $v \in \{0,1\}$ 
		if there is no finite history $H$ applicable to $C$ such that 
		the read operation has returned $1-v$ in $H(C)$.
	\item $C$ is \emph{bivalent} if it is neither $0$-valent nor $1$-valent.  
	\item $C$ is \emph{univalent} if it is \emph{$v$-valent} for exactly one $v \in \{0,1\}$.
	\end{itemize}

	\begin{observationSub}
		\label{obs:Val}
		If the $\oread$ operation has returned some value $v \in \{0,1\}$ in some
		reachable configuration $C$, then $C$ is $v$-valent.
	\end{observationSub}

\begin{proof}
This holds because algorithm $A$ invokes at most one $\oread$ operation.
\end{proof}

	\begin{observationSub}
		\label{obs:BoV}
		A reachable configuration $C$ cannot be both $0$-valent and $1$-valent. 
	\end{observationSub}
	
	\begin{proof}
		Consider the following history: starting from $C$, 
			every process takes steps in a round-robin order, and
			in each step a process receives the earliest pending message for it.
		So all processes take infinitely many steps and every message is eventually received.
		Since the implementation $I$ is $1$-resilient lock-free and no process is correct,
			it is clear that
			both processes $0$ and $1$ complete their execution of algorithm $A$;
			in particular, process $0$ completes its read of the register implemented by~$I$.
		So there exists a finite history $H$ applicable to $C$ such that
			the read operation by process $0$ has completed in $H(C)$.
		Clearly, this read operation returns some value $v \in \{0,1\}$.
		Thus $C$ is not $(1-v)$-valent.
	\end{proof}

	\begin{observationSub}
		\label{obs:BiV}
		If $C$ is a reachable bivalent configuration, then for each $v \in \{0,1\}$, 
			there exists a finite history $H_v$ such that $H_v(C)$ is $v$-valent.
	\end{observationSub}
	
	\begin{proof}
		Let $C$ be a reachable bivalent configuration.
		Thus $C$ is neither $0$-valent nor $1$-valent.
		Thus for each $v \in \{0,1\}$,
			there exists a finite history $H_v$ applicable to $C$ such that
			the $\oread$ operation has returned $v$~in~$H_v(C)$.
		By \autoref{obs:Val}, $H_v(C)$ is $v$-valent.
	\end{proof}
		
	\begin{clm}
		\label{thm:2W1RInit}
		The initial configuration $\CInit$ 
			(before any operation is invoked) is bivalent.
	\end{clm}
	
	\begin{proof} 
		Consider the following run of $A$:
			(a) process $0$ completes its $\textsc{write}(0)$ operation;
			then (b) process $1$ invokes its $\textsc{write}(1)$ operation, completes it, and sends ``OK'';
			then (c) process $0$ receives ``OK'', and completes its $\oread$ operation.
		Clearly, this read returns 1.
		So $\CInit$ is not $0$-valent.

		Consider the following run of $A$:
			(a) process $1$ completes its $\textsc{write}(1)$ operation and sends ``OK'';
			then (b) process $0$ invokes its $\textsc{write}(0)$ operation, and completes it;
			then (c) process $0$ receives ``OK'', and completes its $\oread$ operation.
		Clearly, this read returns 0.
		So $\CInit$ is not $1$-valent.

		Since $\CInit$ is neither $0$-valent nor $1$-valent, it is bivalent.
	\end{proof} 
	
	\begin{clm}
		\label{thm:2W1RCompleteMeansUni}
		If $C$ is a reachable bivalent configuration,
		then no operation has completed in $C$.
	\end{clm}
	
	\begin{proof}
		Suppose, for contradiction,
			that there is a history $H$ applicable to $\CInit$
			such that $H(\CInit)$ is bivalent,
			yet some operation has completed in $H(\CInit)$.
		Suppose the $\oread$ operation has completed in $H(\CInit)$,
			and let $v \in \{0,1\}$  be the value returned.
		Then, by \autoref{obs:Val}, $H(\CInit)$ is $v$-valent --- contradicting that $H(\CInit)$ is bivalent.
		So a write operation has completed in $H(\CInit)$.
		
		Since $H(\CInit)$ is bivalent, it is neither $0$-valent nor $1$-valent.
		Thus for each $v \in \{0,1\}$,
			there exists a finite history $H_v$ applicable to $H(\CInit)$ such that
			the read operation has returned $v$ in $H H_v(\CInit)$.
		For each $v \in \{0,1\}$, let $H'_v = H H_v$.
		
		Let $f$ be any write strong-linearization function for the set of all histories of
		the implementation $I$ of the 2W1R~register.
				
		\begin{subclm}
			\label{scm:Hv}
			For each $v \in \{0,1\}$,
				in $f(H'_v)$, the $\textsc{write}(1-v)$ operation
				occurs before the $\textsc{write}(v)$ operation.\footnote{In an abuse of notation,
				by $f(H'_v)$ we mean $f(H'^I_v)$
				where $H'^I_v$ are
				the steps of the implementation $I$ in $H'_v$.	}
		\end{subclm}
		
		\begin{proof} 
			Recall that in algorithm $A$,
				process $0$ does not invoke its read operation 
				before process $1$ completes its write.
			So, since the read operation has completed in $H'_v(\CInit)$, 
				all three operations have completed in $H'_v(\CInit)$.  
			Thus in $f(H'_v)$, the $\oread$ operation, which returns $v$, must occur after the two $\textsc{write}(-)$ operations.
			Therefore, 
				in $f(H'_v)$, the $\textsc{write}(1-v)$ operation
				occurs before the $\textsc{write}(v)$ operation. 
		\end{proof}

			Since $H$ is a prefix of both $H'_0$ and $H'_1$,
			by Definition~\ref{def:SL},
			$f(H)$ is a prefix of both $f(H'_0)$ and $f(H'_1)$.
		So $f(H)$ does not contain any operations.
		However, since a $\oset$ operation has completed in $H(\CInit)$,
			by Definition~\ref{def-ToS} and~\ref{def:SL}
			$f(H)$ must contain this operation --- a contradiction.

		By \autoref{scm:Hv}, in $f(H'_0)$, the $\textsc{write}(1)$
				occurs before the $\textsc{write}(0)$, and in
				$f(H'_1)$, the $\textsc{write}(0)$
				occurs before the $\textsc{write}(1)$.
		Since $H$ is a prefix of both $H'_0$ and $H'_1$,
			by Definition~\ref{def:WSL},
			the sequence of $\textsc{write}(-)$ operations
			in  $f(H)$ is a prefix of
			the sequence of the $\textsc{write}(-)$ operations
			in both $f(H'_0)$ and $f(H'_1)$.
		So $f(H)$ does not contain any $\textsc{write}(-)$ operation.
		However, since a $\textsc{write}(-)$ operation has completed in $H(\CInit)$,
			by Definition~\ref{def-Reg} and~\ref{def:WSL}
			$f(H)$ must contain this operation --- a contradiction.
	\end{proof}
	
	The rest of the proof is almost the same as 
		in the proof of Theorem~\ref{thm:ToS}.
	
	\begin{clm}
	\label{thm:2W1RNoDeli}
		Let $C$ be any reachable bivalent configuration,
			$e = (p, m)$ be any step that is applicable to $C$,
			$\mathcal{C}_{-e}$ be the set of configurations 
			reachable from $C$ without applying $e$,
			and $\mathcal{D} = e(\mathcal{C}_{-e}) = 
			\{e(E) | E \in \mathcal{C}_{-e} \textrm{ and } 
			e \textrm{ is applicable to } E\}$.
		Then $\mathcal{D}$ contains a bivalent configuration.
	\end{clm}
	
	\begin{proof}
		Note that since $e$ is applicable to $C$, 
			$e$ is also applicable to every configuration in $\mathcal{C}_{-e}$.
			
		Suppose, for contradiction, that 
			$\mathcal{D}$ does not contain a bivalent configuration.
		Then by \autoref{obs:BoV}, $\mathcal{D}$ 
			contains only univalent configurations.
		Since $C$ is bivalent, by \autoref{obs:BiV}, 
			for each $v \in \{0,1\}$ there is a configuration $E_v$
			reachable from $C$ such that $E_v$ is $v$-valent.
		For each $v \in \{0,1\}$,
			if $E_v$ is in $\mathcal{C}_{-e}$,
			then let $E'_v = e(E_v) \in \mathcal{D}$;
			otherwise $e$ was applied to reach $E_v$,
			so let $E'_v$ be a configuration in $\mathcal{D}$
			from which $E_v$ is reachable.
		Thus for each $v \in \{0,1\}$,
			$E'_v$ is in $\mathcal{D}$ and is $v$-valent.
		So $\mathcal{D}$ contains both $0$-valent and $1$-valent configurations.
		
		Consequently, there exists a configuration $\CDeli$ in $\mathcal{C}_{-e}$
			and a step $e' \neq e$ applicable to $\CDeli$ such that
			$e(\CDeli)$ and $e'e(\CDeli)$ are univalent configurations in $\mathcal{D}$
			that have opposite valence. 
		There are two cases: either $e$ and $e'$ are steps of different processes, 
			or $e$ and $e'$ are steps of the same process $p$.
		\begin{description}
			\item[\textbf{Case 1:}] $e$ and $e'$ are steps of different processes.
			
				Then $ee'(\CDeli) = e'e(\CDeli)$.
				Thus, $e'e(\CDeli)$ is both $0$-valent and $1$-valent --- contradicting \autoref{obs:BoV}.
							
			\item[\textbf{Case 2:}] $e$ and $e'$ are steps of the same process $p$.
				Consider the following history: starting from $\CDeli$, 
				every process except $p$ takes steps in a round-robin order, and
			        in each step a process receives the earliest pending message for it.
				So all processes except $p$ take infinitely many steps
				and every message is eventually received.
				Since the implementation $I$ is $1$-resilient lock-free,
				it is clear that
				eventually a configuration is reached where some operation has completed.
				Thus,
					there is a finite $p$-free history $H_{-p}$ applicable to $\CDeli$
					such that an operation has completed in $H_{-p}(\CDeli)$.
				Since $H_{-p}$ is $p$-free,
					$H_{-p}$ is applicable to both $e(\CDeli)$ and $e'e(\CDeli)$,
				and both $e$ and $e'e$ are applicable to $H_{-p}(\CDeli)$.
				Furthermore, it is clear that
					$eH_{-p} (\CDeli) = H_{-p}e(\CDeli)$ and
					$e'eH_{-p} (\CDeli) = H_{-p}e'e(\CDeli)$
					are univalent configurations with opposite valence.
				So $H_{-p}(\CDeli)$ is bivalent ---
					contradicting \autoref{thm:2W1RCompleteMeansUni}
					since an operation has completed in $H_{-p}(\CDeli)$.
		\end{description}
		\vspace*{-3mm}
	\end{proof} 
	
	We now construct an infinite history $\HBi$ applicable to $\CInit$
		such that every configuration in $\HBi(\CInit)$ is bivalent
		as follows:
	\begin{enumerate}
		\item Let $C = \CInit$ ($C$ is bivalent by \autoref{thm:2W1RInit}) and 
			let $S$ be an arbitrary sequence of all $n$ processes.
		\item Let $p$ be the first process in $S$, 
			$m$ be the earliest message in the message buffer of $p$ in $C$ 
			(or $\bot$ if no such message exists), and $e = (p,m)$.
		\item By \autoref{thm:2W1RNoDeli},
			there is a bivalent configuration $C'$ reachable from $C$
			where $e$ has been applied.
		\item Move $p$ to the end of $S$,
			and let $C = C'$.
		\item Repeat from Step 2.
	\end{enumerate}

	Note that all the configurations ``traversed'' by 
		the application of $\HBi$ to $\CInit$ are bivalent.
	By \autoref{thm:2W1RCompleteMeansUni}, no operation has completed in any of them.
	However, in the infinite history $\HBi$, 
		every process takes infinitely many steps and every message is eventually received.
	Thus, since the implementation $I$ of the 2W1R register 
		(used by algorithm $A$) is $1$-resilient lock-free,
		at least one operation by process $0$ or $1$ 
		on this register must have completed
		--- a contradiction.
\end{proof}

\section{Conclusion}
We proved that there is no $1$-resilient lock-free strongly linearizable implementation of the ToS object 
	in asynchronous message-passing systems.
This impossibility result is strong in two dimentions:
	the progress condition is weak
	and the ToS object itself is also weak.
In particular, it implies that there is no $1$-resilient lock-free strongly linearizable implementation of a SWSR register
	in message-passing systems.


%
%
We also proved that there is no $1$-resilient lock-free \emph{write} strongly-linearizable implementation of a 2W1R register
	in asynchronous message-passing systems.
This is in contrast to shared-memory systems,
	since there is a wait-free write strongly-linearizable implementation of MWMR registers from SWMR registers~\cite{HHT211}.


\bibliographystyle{spmpsci}      
\bibliography{arXiv}   

\begin{thebibliography}{1}
\providecommand{\url}[1]{{#1}}
\providecommand{\urlprefix}{URL }
\expandafter\ifx\csname urlstyle\endcsname\relax
  \providecommand{\doi}[1]{DOI~\discretionary{}{}{}#1}\else
  \providecommand{\doi}{DOI~\discretionary{}{}{}\begingroup
  \urlstyle{rm}\Url}\fi

\bibitem{AEW21}
Attiya, H., Enea, C., Welch, J.L.: Impossibility of strongly-linearizable
  message-passing objects via simulation by single-writer registers.
\newblock CoRR \textbf{abs/2105.06614} (2021).
\newblock \urlprefix\url{https://arxiv.org/abs/2105.06614}

\bibitem{FischerEtal1985}
Fischer, M.J., Lynch, N.A., Paterson, M.S.: Impossibility of distributed
  consensus with one faulty process.
\newblock J. ACM \textbf{32}(2), 374--382 (1985).
\newblock \doi{10.1145/3149.214121}.
\newblock \urlprefix\url{http://doi.acm.org/10.1145/3149.214121}

\bibitem{sl11}
Golab, W., Higham, L., Woelfel, P.: Linearizable implementations do not suffice
  for randomized distributed computation.
\newblock In: Proceedings of the Forty-Third Annual ACM Symposium on Theory of
  Computing, STOC ’11, p. 373–382 (2011)

\bibitem{HHT211}
Hadzilacos, V., Hu, X., Toueg, S.: On register linearizability and termination.
\newblock In: Proceedings of the 40th annual ACM Symposium on Principles of
  Distributed Computing, {PODC}, pp. 521--531 (2021)

\bibitem{sl12}
Helmi, M., Higham, L., Woelfel, P.: Strongly linearizable implementations:
  Possibilities and impossibilities.
\newblock In: Proceedings of the 31st annual ACM Symposium on Principles of
  distributed computing, PODC '12, p. 385–394 (2012)

\bibitem{HerlihyWing1990}
Herlihy, M.P., Wing, J.M.: Linearizability: A correctness condition for
  concurrent objects.
\newblock ACM Trans. Program. Lang. Syst. \textbf{12}(3), 463--492 (1990).
\newblock \doi{10.1145/78969.78972}.
\newblock \urlprefix\url{http://doi.acm.org/10.1145/78969.78972}

\end{thebibliography}

\newpage
\appendix

\section{Appendix: Comparing two progress properties}\label{appendix}

We now compare the progress property that we considered in this paper,
	namely, $1$-resilient lock-freedom, and the nonblocking property defined in~\cite{AEW21}.
Recall that:
%

\begin{temptheoremcounter}{OneResilient}
\begin{definition}[1-resilient lock-free]
An implementation is 1-resilient lock-free if it satisfies the following property.
Let $C$ be any reachable configuration that has a pending operation by some process $p$.
For every infinite history $H$
	that is applicable to $C$
	such that at most one process crashes and $p$ is correct,
	there is a finite prefix $H'$ of $H$ such that more operations have completed in $H'(C)$ than in $C$.
\end{definition}
\end{temptheoremcounter}

In \cite{AEW21}, Attiya \emph{et al.} consider object implementations, denoted $I(c,s)$,
	for a set of $c$ \emph{clients} and a disjoint set of $s$ \emph{servers}:
	intuitively, only clients are allowed to invoke object operations, while servers can only help clients perform their operations.
They define the following progress guarantee for such implementations ($g_i$ denotes a configuration): 

\begin{definition} [nonblocking]\label{NonBlocking}\cite{AEW21}
An implementation $I(c,s)$
	is nonblocking iff for every infinite
	execution $e = g_0 \rightarrow ... \rightarrow g_k \rightarrow...$ and $k > 0$, 
	if at least one client and $s-(c-1)$ servers execute a step infinitely often in $e$,
	then some invocation completes after $g_k$ (i.e., the sequence of transitions in $e$ after $g_k$ includes a return transition).
\end{definition}

We first note that nonblocking implementations are 1-resilient lock-free. More precisely: 
 
\begin{observation}\label{obs:one}
For all $c \ge 2$, every nonblocking implementation $I(c,s)$ of an object $O$ 
	is 1-resilient lock-free.
\end{observation}

\begin{proof}
Let $I(c,s)$ be a nonblocking implementation of an object $O$ for some $c\ge2$.
We claim that it is $1$-resilient lock-free. 
Let $C$ be any reachable configuration that has
	a pending operation by some process $p$ (note that $p$ is a client).
So $C=H_0(\CInit)$ for some non-empty finite history $H_0$.
Consider any infinite history $H$
	that is applicable to $C$
	such that at most one process crashes and $p$ is correct.
Let $e$ be the infinite execution that corresponds to the history $H_0 H$ (applied to $\CInit$).
Thus, $e= g_0 \rightarrow ... \rightarrow g_k \rightarrow...$, where $g_0 = \CInit$ and $g_k = C$ for some $k > 0$.
Since (1)~$I(c,s)$ is nonblocking,
	 (2)~a client (namely, $p$) is correct,
	 and
	 (3)~$s-1 \ge s-(c-1)$ servers are also correct,
	 some invocation completes after $g_k$.
Thus, there is a finite prefix $H'$ of $H$ such that more operations have completed in $H'(C)$ than in $C$.
\end{proof}

In contrast to the above, 1-resilient lock-free implementations are not necessarily nonblocking.
To see this, consider a system with $c \ge 2$ clients and $s \ge c$ servers:
	in executions where exactly one client and one server crashes (i.e., there are \emph{two} crashes), nonblocking still guarantees some progress,
	while $1$-resilient lock-freedom does not make any progress guarantees.
In fact:

\begin{observation}\label{obs:two}
There is a $1$-resilient lock-free implementation $I$ of an object $O$ for a system with $n$ processes (consisting of $c=2$ clients
	and $s \ge c$ servers) such that $I$ is \emph{not} nonblocking.
\end{observation}

\begin{proof}
Let $O$ be the trivial object that returns $0$ to all operations.
Let $I$ be the following implementation of $O$ in a system with $n$ processes ($c=2$ clients
	and $s \ge c$ servers).
To invoke an operation on $O$, a client process $p$ sends a message to the other $n-1$ processes.
When a process receives this message, it sends an acknowledgement back to $p$.
Process $p$ waits until it receives $n-2$ acknowledgements and then it returns $0$.
Note that:

\begin{itemize}
\item $I$ is $1$-resilient lock-free because 
	if at most $1$ process crashes, every operation invoked by a correct process eventually completes.

\item $I$ is \emph{not} nonblocking.
To see this, consider an execution $e$ in which one of the two clients and one server crash
	and take no steps (every other process is correct).
In $e$ at least one client and $s-1 = s - (c-1)$ servers take infinitely many steps.
By Definition~\ref{NonBlocking}, nonblocking requires that at least one invocation must complete in execution $e$.
But implementation~$I$ does not satisfy this: since two processes are crashed from the start in $e$,
	the correct client cannot complete any operation (it will never receive $n-2$ acknowledgements)
	so no invocation ever completes in $e$.
\end{itemize}
\vspace*{-3mm}
\end{proof}

Observations~\ref{obs:one} and~\ref{obs:two} imply that
	the nonblocking progress property (Definition~\ref{NonBlocking})
	is strictly stronger than the $1$-resilient lock-free progress property defined here (Definition~\ref{OneResilient}).

\end{document}